\documentclass[11pt]{article}

\usepackage[letterpaper,margin=1.1in]{geometry}
\usepackage{amsthm}
\usepackage{url}

\usepackage{amsmath, amssymb}
\usepackage{cite}
\usepackage{cleveref}
\usepackage{appendix}
\usepackage{graphicx}
\usepackage{epstopdf}
\usepackage{color}
\usepackage{algorithm}
\usepackage[noend]{algpseudocode}
\usepackage{multirow}

\usepackage{caption}
\usepackage{xspace}
\usepackage{mdwlist}

\usepackage{caption}
\usepackage[colorinlistoftodos]{todonotes}

\usepackage{enumitem}
\setitemize{itemsep=2pt, topsep=2pt,parsep=2pt,partopsep=2pt}

\newcommand{\Tabs}{9:    \=11\=11\=11\=11\=11\=11\=11\kill}
\floatstyle{ruled}
\newfloat{alg}{htbp}{alg}
\floatname{alg}{Algorithm}
		{\end{tabbing}\end{algorithm}}

\newcommand{\eps}{\epsilon}

\newcommand{\ceil}[1]{\lceil #1 \rceil}

\newcommand{\threepartdefotherwise}[5]
{
	\left\{
		\begin{array}{ll}
			#1 & \mbox{if } #2 \\
			#3 & \mbox{if } #4 \\
			#5 & \mbox{otherwise}
		\end{array}
	\right.
}

\newtheorem{theorem}{Theorem}[section]
\newtheorem{lemma}[theorem]{Lemma}

\newtheorem{definition}[theorem]{Definition}

\renewcommand{\paragraph}[1]{\vspace{0.15cm}\noindent {\bf #1}}

\algnewcommand\algorithmicswitch{\textbf{switch}}
\algnewcommand\algorithmiccase{\textbf{case}}

\algdef{SE}[SWITCH]{Switch}{EndSwitch}[1]{\algorithmicswitch\ #1\ \algorithmicdo}{\algorithmicend\ \algorithmicswitch}%
\algdef{SE}[CASE]{Case}{EndCase}[1]{\algorithmiccase\ #1}{\algorithmicend\ \algorithmiccase}%
\algtext*{EndSwitch}%
\algtext*{EndCase}%

\algnewcommand\algorithmicwithprob{\textbf{with probability}}
\algnewcommand\algorithmicotherwise{\textbf{otherwise}}

\algdef{SE}[WithProb]{WithProb}{EndWithProb}[1]{\algorithmicwithprob\ #1\ \algorithmicdo}{\algorithmicend\ \algorithmicwithprob}%
\algdef{Ce}[Otherwise]{WithProb}{Otherwise}{EndWithProb}
  {\algorithmicotherwise}%
\algtext*{EndWithProb}%
\algtext*{EndOtherwise}%

\newcommand{\IncludePictures}[1]{}

\newcommand{\withCollision}[1]{}

\newcommand{\FullOrShort}{full}

\ifthenelse{\equal{\FullOrShort}{full}}{
	
	  \newcommand{\fullOnly}[1]{#1}
	  \newcommand{\shortOnly}[1]{}

  }{
  	\newcommand{\fullOnly}[1]{}
	  \newcommand{\shortOnly}[1]{#1}
  }

\begin{document}






\date{}

\title{Fast Structuring of Radio Networks\\[0.2cm] \Large for Multi-Message Communications}

\author{
Mohsen Ghaffari\\
MIT\\
\texttt{ghaffari@mit.edu}
\and
Bernhard Haeupler\\
Microsoft Research\\
\texttt{haeupler@cs.cmu.edu}}

\maketitle

\begin{abstract}
We introduce collision free layerings as a powerful way to structure radio networks. These layerings can replace hard-to-compute BFS-trees in many contexts while having an efficient randomized distributed construction. We demonstrate their versatility by using them to provide near optimal distributed algorithms for several multi-message communication primitives.
\medskip

Designing efficient communication primitives for radio networks has a rich history that began 25 years ago when Bar-Yehuda et al. introduced fast randomized algorithms for broadcasting and for constructing BFS-trees. Their BFS-tree construction time was $O(D \log^2 n)$ rounds, where $D$ is the network diameter and $n$ is the number of nodes. \ \ Since then, the complexity of a broadcast has been resolved to be $T_{BC} = \Theta(D \log \frac{n}{D} + \log^2 n)$ rounds. On the other hand, BFS-trees have been used as a crucial building block for many communication primitives and their construction time remained a bottleneck for these primitives.

\medskip

We introduce collision free layerings that can be used in place of BFS-trees and we give a randomized construction of these layerings that runs in nearly broadcast time, that is, w.h.p. in $T_{Lay} = O(D \log \frac{n}{D} + \log^{2+\eps} n)$ rounds for any constant $\eps>0$. We then use these layerings to obtain: \ \ (1) A randomized algorithm for gathering $k$ messages running w.h.p. in $O(T_{Lay} + k)$ rounds. \ \ (2) A randomized $k$-message broadcast algorithm running w.h.p. in $O(T_{Lay} + k \log n)$ rounds. These algorithms are optimal up to the small difference in the additive poly-logarithmic term between $T_{BC}$ and $T_{Lay}$. Moreover, they imply the first optimal $O(n \log n)$ round randomized gossip algorithm. 
\end{abstract}

\newpage
\section{Introduction}
Designing efficient communication protocols for radio networks is an important and active area of research. Radio networks have two key characteristics which distinguish them from wired networks: For one, the communications in these networks have an inherent broadcast-type nature as the transmissions of one node can reach all nearby nodes. On the other hand, simultaneous transmissions interfere and this interference makes the task of designing efficient communication protocols challenging. A standard model that captures these characteristics is the \emph{radio networks model}~\cite{CK}, in which the network is abstracted as a graph $G=(V,E)$ with $n$ nodes and diameter $D$. Communication occurs in synchronous rounds, where in each round, each node either listens or transmits a message with bounded size. A node receives a message if and only if it is listening and exactly one of its neighbors is transmitting. Particularly, a node with two or more transmitting neighbors cannot distinguish this collision from background noise. That is, the model assumes \emph{no collision detection}. 

Communication problems in radio networks can be divided into two groups: single-message problems like single-message broadcast, and multi-message problems such as $k$-message broadcast, gossiping, $k$-message gathering, etc. By now, randomized single-message broadcast is well-understood, and is known to have asymptotically tight time-complexity of $T_{BC}=\Theta(D \log \frac{n}{D} + \log^2 n)$ rounds~\cite{CR, KP03, ABLP, KM}\footnote{We remark that, throughout the whole paper, when talking about randomized algorithms, we speak of the related time-bound that holds with high probability (w.h.p), where w.h.p. indicates a probability at least $1-\frac{1}{n^\beta}$ for an arbitrary constant $\beta\geq 2$.}. On the other hand, multi-message problems still remain challenging. The key issue is that, when aiming for a time-efficient protocol, the transmissions of different messages interfere with each other. Bar-Yehuda, Israeli and Itai \cite{BII93} presented an $O(D\log^2 n)$ round construction of Breadth First Search trees and used this structure to control the effects of different messages on one another in multi-message problems. Since then, BFS trees have become a standard substrate for multi-message communication protocols (see, e.g., \cite{CGL, KK, CKR, GPX05}). However, the best known construction for BFS trees remains $O(D \log^2 n)$ and this time-complexity has become a bottleneck for many multi-message problems. 

\subsection{Our Results}
As the main contribution of this paper we introduce collision-free layering which are simple node numberings with certain properties (see \Cref{sec:layer} for definitions). Layerings are structures that can be viewed as relaxed variants of BFS trees and can replace them in many contexts while having an efficient randomized construction. We present a randomized construction of these layerings that runs in $T_{Lay} = O(D \log \frac{n}{D} + \log^{2+\eps} n)$ rounds for any constant $\eps>0$. This round complexity is almost equal to the broadcast time, i.e., $T_{BC}=\Theta(D \log \frac{n}{D} + \log^{2} n)$ rounds, and is thus near-optimal. 

\smallskip
Using collision free layerings, and with the help of additional technical ideas, we achieve the following near-optimal randomized algorithms for the aforementioned multi-message problems:
\begin{itemize}
\item[(A)] A randomized algorithm for $k$-message single-destination gathering that with high probability gathers $k$ messages in $O(T_{Lay} + k)$ rounds.
\item[(B)] A randomized algorithm for $k$-message single-source broadcast with complexity $O(T_{Lay} + k \log n)$ rounds, w.h.p. This algorithm uses network coding.
\item[(C)] The above algorithms also lead to the first optimal randomized all-to-all broadcast (gossiping) protocol, which has round complexity $O(n\log n)$ rounds\footnote{We remark that an $O(n \log n)$ gossiping solution was attempted in~\cite{FQ06}, for the scenario of known topology, but its correctness was disproved~\cite{QinPersonalCommunication}.}.
\end{itemize}

\smallskip
Note that modulo the small difference between $T_{Lay}$ and $T_{BC}$, the time complexity of the above algorithms are optimal, and that they are the first to achieve the optimal dependency on $k$ and $D$.

\subsection{Related Work}
Communication over radio networks has been studied extensively since the 70's. In the following, we present a brief overview of the known results that directly relate to the setting studied in this paper. That is, randomized algorithms\footnote{We remark that typically the related deterministic algorithms have a different flavor and incomparable time-complexities, with $\Omega(n)$ often being a lower bound.}, with focus on with high probability (whp) time and under the standard and least demanding assumptions: without collision detection, unknown topology, and with messages of logarithmic size.

\medskip
\noindent\textbf{Single-Message Broadcast:} Bar-Yehuda, Goldreich, and Itai (BGI)~\cite{BGI1} gave a simple and efficient algorithm, called Decay, which broadcasts a single message in $O(D \log n +\log^2 n)$ rounds. Alon et al.~\cite{ABLP} proved an $\Omega(\log^2 n)$ lower bound, which holds even for centralized algorithms and graphs with constant diameter. Kushilevitz and Mansour~\cite{KM} showed an $\Omega(D \log{\frac{n}{D}})$ lower bound. Finally, the remaining gap was closed by the simultaneous and independent algorithms of \cite{CR} and\cite{KP03}, settling the time complexity of single-message broadcast to $T_{BC}=\Theta(D\log{\frac{n}{D}}+\log^2 n)$. 

\medskip
\noindent\textbf{$k$-Message Gathering and $k$-Unicasts:} Bar-Yehuda, Israeli and Itai (BII)~\cite{BII93} presented an algorithm to gather $k$ messages in a given destination in whp time $O(k \log^2 n +D \log^2 n)$, using the key idea of routing messages along a BFS tree via Decay protocol of \cite{BGI1}. The bound was improved to $O(k\log n + D\log^2 n)$ ~\cite{CKR} and then to $O(k + D\log^2 n)$~\cite{KK}, using the same BFS approach but with better algorithms on top of the BFS. A deterministic $O(k \log n + n \log n)$ algorithm was presented in \cite{CGL}, which substitutes the BFS trees with a new concept of Breadth-Then-Depth. 


\medskip
\noindent\textbf{$k$-Message Broadcast:} BII~\cite{BII93} also used the BFS-based approach to broadcast $k$-message in whp time $O(k\log^2 n + D\log^2 n + \log^3 n)$. Khabbazian and Kowalski~\cite{KK} improve this to $O(D\log^2 n + k\log n + \log^3 n)$ using network coding. Ghaffari et al.~\cite{NCLB} showed a lower bound of $\Omega(k \log n)$ for this problem, even when network coding is allowed, which holds even for centralized algorithms. 

\medskip
\noindent\textbf{Gossiping}: Gasieniec~\cite{Gossip10} provides a good survey. The best known results are $O(n \log^2 n)$ algorithm of Czumaj and Rytter~\cite{CR} and the $\Omega(n\log n)$ lower bound of Gasieniec and Potapov~\cite{GP}. The lower bound holds for centralized algorithms and also allows for network coding. Same can be inferred from~\cite{NCLB} as well. An $O(n \log n)$ algorithm was attempted in~\cite{FQ06}, for the scenario of known topology, but its correctness was disproved~\cite{QinPersonalCommunication}.



\section{Preliminaries}

\subsection{The Model}
We consider the standard \emph{radio network model}\cite{CK, BGI1}: The network is represented by a connected graph $G = (V, E)$ with $n=|V|$ nodes and diameter $D$. Communication takes place in synchronous rounds. In each round, each node is either listening or transmitting a packet. In each round, each listening node that has exactly one transmitting neighbor receives the packet from that neighbor. Any node that is transmitting itself or has zero or more than one transmitting neighbor does not receive anything. In the case that two or more neighbors of a listening node $v \in V$ are transmitting, we say a \emph{collision} has happened at node $v$. We assume that each transmission (transmitted packet) can contain at most one message as its \emph{payload} plus an additive $\Theta(\log n)$ bits as its \emph{header}. Since we only focus on randomized algorithms, we can assume that nodes do not have original ids but each node picks a random id of length $4\log n$ bits. It is easy to see that, with high probability, different nodes will have different ids.

\subsection{The Problem Statements}
We study the following problems:

\begin{itemize}
\item\textbf{$k$-message Single-Destination Gathering:} $k$ messages are initially distributed arbitrarily in some nodes and the goal is to gather all these messages at a given \emph{destination} node. 
\item\textbf{Single-Source $k$-Message Broadcast:} A single given \emph{source} node has $k$ messages and the goal is to deliver all messages to all nodes.
\item\textbf{Gossiping:} Each node has a single message and the goal is for each node to receive all messages. 
\end{itemize}
%
In each problem, when stating a running time for a randomized algorithm, we require that the algorithm terminates and produces the desired output within the stated time with high probability (in contrast to merely in expectation). 

We make the standard assumptions that nodes do not know the topology except a constant-factor upper bound on $\log n$. From this, given the algorithms that we present, one can obtain a constant factor estimation of $D$ and $k$ using standard double-and-test estimation techniques without more than a constant factor loss in round-complexity. We skip these standard reductions and assume that constant-factor approximations of $D$ and $k$ are known to the nodes. For simplicity, we also assume that $k$ is at most polynomial in $n$. 

\subsection{A Black-Box Tool: The CR-Braodcast Protocol}\label{sec:decaynew}
Throughout the paper, we make frequent use of the optimal broadcast protocol of Czumaj and Rytter (CR)~\cite{CR}. Here, we present a brief description of this protocol. To describe this protocol, we first need to define a specific infinite sequence of positive integers $BC$ with the following properties:

\begin{enumerate}
	\item[(1)] Every consecutive subsequence of $\Omega(\log \frac{n}{D})$ elements in $BC$ contains $1,2,\ldots,\log \frac{n}{D}$ as a subsequence.
	\item[(2)]	For every integer $k \in [\log \frac{n}{D}, \log \frac{n}{D} + \log \log n]$, any consecutive subsequence of $\Omega(\log \frac{n}{D} \cdot 2^k)$ elements in $BC$ contains an element of value $k$.
	\item[(3)] Every consecutive subsequence of $\Omega(\log n)$ elements in $BC$ contains $1,2,\ldots,\log n$ as a subsequence.
\end{enumerate}

These properties were defined in\cite[Definition 7.6]{CR} under the name \emph{$D$-modified strong deterministic density property}\footnote{We remark that the Property 3 stated here is slightly stronger than the property 3 of \cite[Definition 7.6]{CR}, but is satisfied by the sequence provided in \cite{CR}. This modification is necessary to achieve the $k \log n$ dependence on number of messages $k$ in the $k$-message broadcast problem\Cref{sec:bcast}. Using the original definition would lead to a time bound of $\Omega(k \log n \log \frac{n}{D})$}. Furthermore, it can be easily verified that the following sequence, which is again taken from \cite{CR}, satisfies these properties.  

\bigskip
\noindent \fbox {\parbox{\textwidth}{\vspace{0.05cm}
For any $n$ and $D$, we define the sequence $BC = BC_0, BC_1, \ldots$ such that for each non-negative integer $j$, we have:
\begin{center} 
$BC_{3j} = \log \frac{n}{D} + k$, where $k$ is such that $(j \mod \log n) \equiv 2^k \mod 2^{k+1}$\\
$BC_{3j+1} = j \mod \log \frac{n}{D}$ and\\
$BC_{3j + 2} = j \mod \log n$.
\end{center}
}}
\bigskip

We now present the pseudo-code of the broadcast protocol of \cite{CR}, which will be used throughout the rest of the paper. This protocol has $4$ key parameters: two disjoint sets $A$, $R$ and two integer values $\delta$ and $T$. It is assumed that each node $v$ knows the values of $\delta$ and $T$ and it also knows whether it is in $A$ and $R$, via Boolean predicates of the form $(v \in A)$ and $(v \in R)$. Each node $v \in A$ has a message $\mu_v$ (which is determined depending on the application of the protocol). The protocol starts with nodes in $A$ where each active node $v\in A$ forwards its message. The nodes in $R$ become active (join $A$) at the end of the first phase in which they receive a message, and retransmit this message in the next phases. \Cref{alg:decay} presents the pseduo-code for algorithm CR-Broadcast($A$, $R$, $\delta$, $T$):


\begin{algorithm}[h]
\caption{Algorithm CR-Broadcast($A$, $R$, $\delta$, $T$) @ node $v$:}
\begin{algorithmic}[1]
\normalsize
\Statex Syntax: each \textsc{transmit} or \textsc{listen} corresponds to one communication round
\Statex
\If{$(v \in A)=false$} $\mu_v \gets \emptyset$ \EndIf

\For {phase $i = 1$ to $T$}
	\For {$j = 1$ to $\delta$}
		\If {$(v\in A)=true$}
			\WithProb{$2^{-BC_{i \delta + j}}$}
				\State \textsc{transmit} $(v.id, \mu_v)$
			\Otherwise
			\State \textsc{listen}
			\EndWithProb
		\Else
	    \State  \textsc{listen}
		\EndIf
		\If {received a message $(u.id, \mu)$} $\mu_v \gets \mu$ \EndIf
	\EndFor
	\If {$\mu \neq \emptyset$ \& $(v\in R)$} $(v \in A) \gets true$ \EndIf
\EndFor
\label{alg:decay}
\end{algorithmic}
\end{algorithm}

%
%
%
%
%
%
%
%
%
%
%
%
%

We will use the following lemma from \cite{CR} and \cite{BGI1}:

\begin{lemma}\label{lem:CR-global}
For any connected network $G = (V,E)$ with diameter $D$ and for any node $v$, an execution of $\mu_v$ CR-Broadcast($\{v\}$, $V \setminus \{v\}$, $\delta$, $T$) with $T = \Theta(D (\log \frac{n}{D} + \delta) + \log^2 n)/\delta$ leads with high probability to $S_0 = V$ and $\mu_u = \mu_v$. That is, broadcasting a message from $v$ to all nodes takes with high probability at most $T$ rounds. 
\end{lemma}

\begin{lemma}\label{lem:CR-local}
In each execution of CR-Broadcast protocol, for any two neighboring nodes $u$ and $v$, if $(u \in A)=true$ and $(v \in A)=false$ at round $r$, then in round $r+\Theta(\log^2 n)$, w.h.p., node $v$ has received a message from some node.
\end{lemma}

\section{Layerings}\label{sec:layer}
Here, we introduce \emph{layerings} and we provide a set of algorithms for constructing layerings with desirable properties. 
\subsection{Definitions}
In short, layerings are particular types of numbering of nodes; they organize and locally group nodes in a way that is useful for multi-message gathering and broadcasting tasks and for parallelzing and pipelining communications. In this subsection, we present the formal definitions.
\begin{definition}(\textbf{layering})\label{def:layering}
A layering $\ell$ of graph $G=(V,E)$ assigns to each node $u \in V$ an integer layer number $\ell(u)$ such that \textbf{(a)} there is only one node $s$ with $\ell(s)=0$, known as the \emph{source}; and \textbf{(b)} every node $u$, except the source, is connected to a node $v$ such that $\ell(v) < \ell(u)$. We define the depth of layering $\ell$ to be equal to $\max_{u\in V} \ell(u)$. In the distributed setting, for a layering $\ell$, we require each node $u$ to know its layer number $\ell(u)$, and also, for each node $u$ other than the source, we require $u$ to know (the ID of) one node $v$ such that $\ell(v) <\ell(u)$ and $u$ is a neighbor of $v$. In this case, we call $v$ \emph{the parent} of $u$. 
\end{definition}

\begin{definition} (\textbf{$C$-collision-free layering}) A layering $\ell$ together with a $C$-coloring of the nodes $c:V \rightarrow\{0,\ldots,C-1\}$ is \emph{$C$-collision-free} if for any two nodes $u$ and $v$ such that $\ell(u) \neq \ell(v)$ and $dist_{G}(u,v) \leq 2$, we have $c(u) \neq c(v)$. In the distributed setting, we require each node $v$ to know the value of $C$ and also its own color $c(v)$.
\end{definition}

\begin{definition} (\textbf{$d$-stretch layering}) A layering $\ell$ is \emph{$d$-stretch} if for any two neighboring nodes $u$ and $v$, we have $|\ell(u) - \ell(v)|\leq d$.
\end{definition}

\begin{figure}[t]
	\centering
		\includegraphics[width=0.75\textwidth]{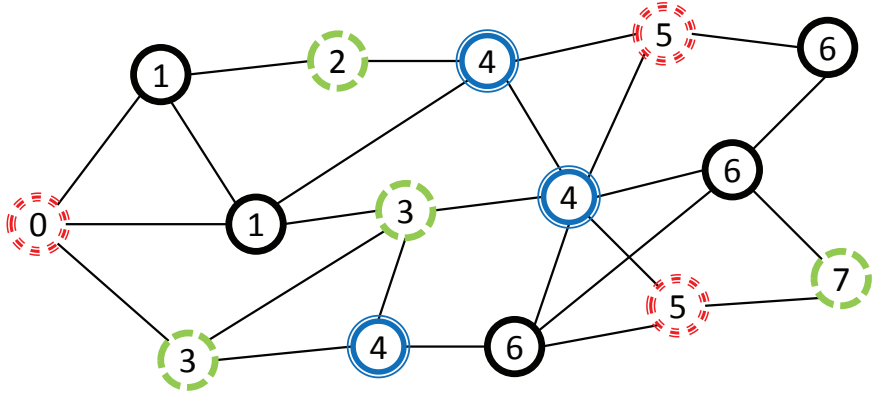}
	\caption{\small A $4$-collision-free layering with depth $7$ and stretch $3$. The number in each node indicates its layer number.}
	\label{fig:Layering}
\end{figure}


We remark that a \emph{BFS-layering} in which each node is labeled by its distance from the source is a simple example for a layering with stretch $1$ and depth $D$. We also remark that any $d$-stretch layering $\ell$ can also be made $(2d+1)$-collision-free by choosing $C = 2d+1$ and $c(u) = \ell(u) \mod C$. This makes BFS-layerings $3$-collision free. In the next sections we show that \emph{pseudo-BFS layerings}, that is, layerings with similar collision freeness and depth, can be constructed efficiently and can replace BFS layerings in many scenarios:

%
%
%
%

\begin{definition}(\textbf{pseudo-BFS layering})
A layering (and a related coloring) is a pseudo-BFS layering if it is $O(1)$-collision-free and has depth $O(D + \log n)$. 
\end{definition}



\subsection{Layering Algorithms}

Here, we show that pseudo-BFS layerings can be constructed in almost broadcast time, that is, $T_{BC} = O(D\log\frac{n}{D}+\log^{2} n)$ rounds. This is faster than the best known construction time of BFS layerings, which remains $O(D \log^2 n)$ rounds.

\begin{theorem}\label{thm:layer} There is a distributed randomized algorithm that for any constant $\epsilon>0$, constructs a pseudo-BFS layering w.h.p., in $O(D \log{\frac{n}{D}} + \log^{2+\epsilon} n)$ rounds. 
\end{theorem}


\subsubsection{Starter: A construction with round-complexity $O(D \log{\frac{n}{D}} + \log^{3} n)$}
\begin{theorem}\label{crl:layer} There is a distributed randomized algorithm that w.h.p. constructs a pseudo-BFS layering from a given source node $s$ in $O(D \log{\frac{n}{D}} + \log^3 n)$ rounds.
\end{theorem}
The high-level outline of this construction is to start with a crude basic layering obtained via a broadcast and then refining this layering to get a pseudo-BFS layering.
Given the broadcast protocol presented in \Cref{sec:decaynew}, we easily get the following basic layerings:

\begin{lemma}\label{lem:basicLayering}
For any $\delta \in [\log \frac{n}{D},\log^2 n]$ there is a layering algorithm that  computes, w.h.p., an $O(D + \frac{\log^2 n}{\delta})$-depth layering with a given source $s$ and stretch $O(\frac{\log^2 n}{\delta})$ in $O(D \delta + \log^2 n)$ rounds.
\end{lemma}
\begin{proof} We run the CR-Broadcast algorithm with parameter $\delta$, $T=\Theta(D (\log \frac{n}{D} + \delta) + \log^2 n)/\delta$, $A={s}$ and $R=V\setminus {s}$. For each non-source node $v$ we then set $\ell(v)$ to be the smallest phase number in which $v$ receives a message, and the parent of $v$ to be the node $w$ from which $v$ receives this first message. \Cref{lem:CR-global} guarantees that indeed after $T\delta$ rounds all nodes are layered. The depth of the layering can furthermore not exceed the number of iterations $T = \Theta(D \delta + \log^2 n)/\delta$. The stretch part of the lemma follows from \Cref{lem:CR-local} which guarantees that two neighboring nodes receive their messages at most $O(\log^2 n)$ rounds and therefore at most $O(\frac{\log^2 n}{\delta})$ iterations apart. 
\end{proof}

Next we give the algorithm to refine the basic layerings of \Cref{lem:basicLayering} to a pseudo-BFS layering. \shortOnly{We present the algorithm but defer the correctness proof to the full version.} 

\begin{lemma}\label{thm:refine} Given a $d$-stretch layering $l$ with depth $D'$, the Layer Refinement Algorithm (LRA) computes a $5$-collision-free $O(d)$-stretch layering $l'$ with depth $O(D')$ in $O(d \log^2 n)$ rounds. 
\end{lemma}
\fullOnly{\begin{proof}\shortOnly{[Proof of \Cref{thm:refine}]}
We first show that the $l'(.)$ layering achieved by the LRA algorithm is $5$-collision-free. Then, we argue that the stretch of $l'(.)$ is $O(d)$, and its depth is $O(D')$.

For the first part, we show that with high probability, for any two nodes $u$ and $v$ such that $l'(u) \neq l'(v)$ but $c(u) =c(v)$, the distance of $u$ and $v$ is at least $3$. Suppose $u$ and $v$ are such that $l'(u) \neq l'(v)$ but $c(u) \equiv c(v) \equiv a \pmod{5}$ for some $a \in \{0, 1, 2, 3, 4\}$. We first show the statement for the case where $a = 0$, i.e., when $u$ and $v$ are boundary nodes. Suppose the first message that $v$ and $u$ received in the boundary detection part were from nodes $w_1$ and $w_2$ respectively. Then we know that $\lceil \frac{l(w_1)}{d} \rceil \neq \lceil \frac{l(w_2)}{d} \rceil$. This is because, otherwise there would exist a $j$ such that $l(v), l(u) \in [5jd-d, 5jd]$ and thus, $l'(u) = l'(v)$ which is by assumption not the case. Since $\lceil \frac{l(w_1)}{d} \rceil \neq \lceil \frac{l(w_2)}{d} \rceil$, and $\lceil \frac{l(w_1)}{d} \rceil \equiv \lceil \frac{l(w_2)}{d} \rceil \equiv 1 \pmod{5}$, we can infer that $|l(w_1) - l(w_2)| > 4d$. Thus, $|l(v) - l(u)| >3d$. Noting that $l$ is a $d$-stretch layering, we conclude that the distance between $u$ and $v$ is greater than $3$.

Now consider the case where $a=1$, i.e., when $u$ and $v$ are start-line nodes. This case is similar to the $a=0$ case. In particular, we know that $\lceil \frac{l(u)}{d} \rceil \neq \lceil \frac{l(v)}{d} \rceil$. 
This is because otherwise, $u$ and $v$ would get the same $l'$ layer number, which would be a contradiction. Therefore, we can infer that $|l(u) - l(v)|>4d$, which shows that the distance between $u$ and $v$ is greater than $3$.

Finally, consider the case where $a \in \{2,3,4\}$. If $u$ and $v$ are in two different strips, then their distance is at least $3$ as they are separated at least by one boundary layer and one start-line layer. Suppose $u$ and $v$ are in the same strip. Since from \Cref{lem:CR-local} we know that the CR-Broadcast protocol with parameter $\delta=\Theta(\log^2 n)$ makes exactly one hop of progress in each phase, and as we cycle over colors $\{2,3,4\}$, the distance between two nodes of the same color in the same strip is at least $3$. 

\medskip
For the second part, we show that with high probability, the stretch of the $l'$-layering achieved by the LRA is at most $10d$ and it has depth $O(D')$. For the depth claim, note that the largest possible $l'$ layer number for boundary nodes is at most $2d(\lceil \frac{D'}{d} \rceil+1)+5d \leq 2D'+7d =O(D')$. For the stretch part, note that the difference between the $l$ layers of two consecutive boundary $l'$-layers is exactly $10d$. 
Now note that, any two neighboring nodes are within two consecutive boundary layers (including the boundary layers themselves). Thus, the difference between $l'$ layers of each two neighbors is at most $10d$ which means that the stretch of layering $l'$ is at most $10d$.
\end{proof}
}

\paragraph{Layer Refinement Algorithm (LRA):} Throughout the presentation of the algorithm, we refer to Figure \ref{fig:LRA} as a helper tool and also present some intuitive explanations to help the exposition. 

As the first step of the algorithm, we want to divide the problem into small parts which can be solved in parallel. For this purpose, we first run the CR-broadcast protocol with parameters $T=1$, $\delta=\Theta(\log^2 n)$, $A$ equal to the set of nodes $u$ such that $\lceil \frac{l(u)}{d} \rceil \equiv 1 \pmod{5}$, and $R=\emptyset$. Each node $u \in A$ sets message $\mu_u$ equal to $l(u)$. In Figure \ref{fig:LRA}, these nodes are indicated by the shaded areas of width $d$ layers. Since layering $l$ has stretch at most $d$, each shaded area cuts the graph into two non-adjacent sets, above and below the area (plus a third part of the shaded area itself). After these transmissions, each node $v$ becomes a \emph{boundary node} if during these transmissions, $v$ was not transmitting but it received a message from a node $w$ such that $l(w)> l(v)$. In Figure \ref{fig:LRA}, the boundary nodes are indicated via red contour lines. These boundaries divide the problem of layering into strips each containing at most $5d$ layers, and such that two nodes at different strips are not neighbors. For each boundary node $v$, we set $l'(v) = 2d(\lceil \frac{l(v)}{d} \rceil+1)$ and color it with color $0$, i.e., $c(v)=0$.

\begin{figure}[t]
	\centering
		\includegraphics[width=0.65\textwidth]{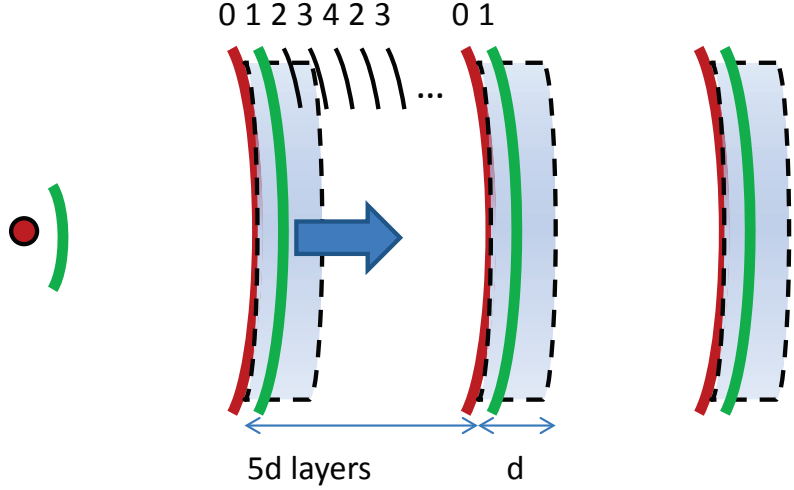}
	\caption{Layer Refinement}
	\label{fig:LRA}
\end{figure}

Next, we indicate the direction starting from the boundary which moves in the increasing direction of layer numbers $l$. For this, we run the CR-broadcast protocol with parameters $T=1$, $\delta=\Theta(\log^2 n)$, $A$ equal to the set of boundary nodes, and $R =\emptyset$, where each boundary node $u$ sets $\mu_u$ equal to $(l(u), l'(u))$. A non-boundary node $v$ that receives a message from a boundary node $w$ such that $l(w)<l(v)$ is called a \emph{start-line} node. In Figure \ref{fig:LRA}, start-line nodes are indicated via green contour lines. Every such node $v$ sets it $l'$-layer number to $l'(v) = l'(w)+1$ and its color $c(v)=1$, and records $w.id$ as the id of its parent.

Next, we assign $l'$ layer numbers to nodes inside the strips, starting from the start-line nodes and moving upwards in ($l$-layer numbers) till reaching the next layer of boundary nodes. This is done for different strips in parallel, using the CR-Broadcast protocol with parameter $T=5d$, $\delta=\Theta(\log^2 n)$, $A$ equal to the set of start-line nodes and $R$ equal to the set of nodes that are neither boundary nor start-line. As a result, in each phase of the CR-Broadcast, all non-boundary nodes that have received an $l'$ layer number by the start of that phase try transmitting their $l'$ layer number and their id. In every phase, a node $v$ that does not have an $l'$ layer number yet and receives a transmission from a node $w$ records $w.id$ as the id of its parent and sets its $l'$-layer number $l'(v)=l'(w)+1$ and $c(v)=2+((c(w)+1) \mod{3})$. 
In other words, the color number is incremented every time modulo $5$, but skipping colors $0$ and $1$ (preserved respectively for boundary and start-line nodes).  In Figure \ref{fig:LRA}, the numbers at the top part indicate these color numbers. From \Cref{lem:CR-local}, we get that the wave of the layering proceeds exactly one hop in each phase. Since in each phase, only nodes that do not have an $l'$ layer get layered, the waves of layering stop when they reach boundary nodes. Finally, each boundary node $v$ records the id of the node $w$ from which $v$ hears the first message as the id of its parent.

Next, we present the proof of \Cref{crl:layer} which uses the Layer Refinement Algorithm (LRA) on top of the basic layering provided by \Cref{lem:basicLayering}.

\begin{proof}[Proof of \Cref{crl:layer}] If $D< n^{0.1}$, we construct a basic layering with stretch $O(\log n)$ and depth $O(D+ \log n)$ in $O(D\log n + \log^2 n)$ rounds by using \Cref{lem:basicLayering} with parameter $\delta=\Theta(\log n)$. Then, we use the LRA to get to an $O(1)$-collision-free layering with depth $O(D + \log n)$ in additional $O(\log^3 n)$ rounds (\Cref{thm:refine}). The total round complexity becomes $O(D\log n + \log^3 n) = O(D \log \frac{n}{D} + \log^3 n)$.

If $D\geq n^{0.1}$, we construct a basic layering with stretch $O(\log^2 n/\delta)=O(\log^2 n)$ and depth $O(D + \log^2 n / \delta) = O(D)$ in $O(D\log \frac{n}{D} + \log^2 n)$ rounds by using \Cref{lem:basicLayering} with parameter $\delta = \log{\frac{n}{D}}$. Then, we use the LRA to get to an $O(1)$-collision-free layering with depth $O(D)$, in additional $O(\log^4 n)$ rounds. The total round complexity becomes $O(D\log \frac{n}{D} + \log ^2 n +  \log^4 n) = O(D \log{\frac{n}{D}})$.

In both cases the round complexity is $O(D \log{\frac{n}{D}} + \log^3 n)$ and the depth is $O(D + \log n)$.
\end{proof}

\subsubsection{Reducing the round complexity to $O(D \log{\frac{n}{D}} + \log^{2+\eps} n)$} 
The construction time in \Cref{crl:layer} is asymptotically equal to the broadcast time $T_{BC}$, for all values of $D =\Omega(\log^2 n)$. Here we explain how to achieve an almost optimal round complexity for smaller $D$ by reducing the pseudo-BFS construction time to $O(D \log{\frac{n}{D}} + \log^{2+\eps} n)$ rounds, for any constant $\eps>0$. 

\medskip\noindent\textbf{Recursive Layering Refinement Algorithm:} In the LRA algorithm, we used the CR-Broadcast protocol with parameter $\delta=O(\log^2 n)$ to refine the layering numbers inside each strip, in $O(\log^3 n)$ rounds. The key change in RLRA is that, we perform this part of refinement in a faster manner by using a recursive refinement algorithm with $O(1/\eps)$ recursion levels. We remark that, this speed-up comes at a cost of a $2^{O(1/\eps)}$ factor increase in the depth and $O(1/\epsilon)$ factor increase in the round complexity, and also in using $O(1/\epsilon)$ colors (instead of just $5$), for the final layering. However, since we assume $\eps$ to be constant, these costs do not effect our asymptotic bounds. 

Let $r=\ceil{1/\epsilon}$ and $\tau=\alpha \log^{\frac{1}{r}} n$ for a sufficiently large constant $\alpha$. In the $i^{th}$ level of recursion, we get an algorithm $A_{i}$ that layers a graph with depth $\tau^i$ using $2i+1$ colors, in $i \cdot \Theta(\log^{2+\frac{1}{r}})$ rounds. 

For the base case of recursion, algorithm $A_{1}$ is simply using the CR-Broadcast algorithm with parameter $\delta=\Theta(\log^2 n)$, and $T= \tau$ phases. Then, we assign layer numbers $\ell_1()$ based on the phase in which each node receives its first message, and set $c(v)=\ell_1(v) \pmod{3}$.

We get algorithm $A_{i}$ using algorithm $A_{i-1}$ as follows: First, use the CR-Broadcast algorithm with parameter $\delta=\Theta(\log^{2-\frac{i-1}{r}})$ and $T=\tau^i$ phases. From this broadcast, we get a layering $\ell^*$ that has stretch at most $d_i=\Theta(\log^{\frac{i-1}{r}} n) \leq \delta^{i-1}/5$. Then, using this layering, similar to the LRA, we break the graph into $\Theta(\delta)$ strips which each contain $\Theta(\delta^{i-1})$ layers. It is easy to see that, each strip has depth at most $\Theta(\delta^{i-1})$. Next, we determine boundary and start-line nodes as in the LRA and layer and color them. In particular, we assign color $2i+1$ to the boundaries of these strips and set their layer number $\ell_{i}(v)=2\delta^{i-1}(\ceil{\frac{\ell^*(v)}{\delta^{i-1}}}+1)$. Moreover, we assign color $2i$ to the start-lines of these strips and layer each start-line node $v$ with $l_{i}(v)=l_{i}(w)+1$, where $w$ is the first boundary node from which $v$ receives a message. Inside each strip, which is a graph with depth $\delta^{i-1}$, we use algorithm $A_{i-1}$ with colors $1$ to $2(i-1)+1=2i-1$. 

Following $r$ recursion steps, we get algorithm $A_{r}$, which layers a graph with depth $\tau^r = \Theta(\log n)$ using $2r+1=O(r)$ colors, in $r \cdot \Theta(\log^{2+\frac{1}{r}} n) = \Theta(\log^{2+\epsilon} n)$ rounds. In the LRA, if we substitute the part that layers each strip in $\Theta(\log^3 n)$ rounds with $A_{r}$, we get the recursive layering refinement algorithm.


\begin{proof}[Proof of \Cref{thm:layer}] For the case where $D\geq n^{{0.1}}$, we simply use the LRA algorithm and calculations are as before. For the case where $D < n^{0.1}$, the proof is similar to that of \Cref{crl:layer} with the exception of using the Recursive Layering Refinement Algorithm instead of the LRA. 
\end{proof}

%

\section{Gathering}\label{sec:gather}
In this section, we present a $k$-message gathering algorithm with round complexity $O(T_{Lay} + k)$. This round complexity is near optimal as $k$-message gathering has a lower bound of $T_{BC}+ k$. The additive $k$ term in this lower bound is trivial. The $T_{BC}$ term is also a lower bound because the lower bounds of single-message broadcast extend to single-message unicast from an adversarially chosen source to an adversarially chosen destination, and single-message uni-cast is a special case of $k$-message gathering where $k$=1.

\begin{theorem}\label{thm:gathering-total}
There is a distributed randomized algorithm that, w.h.p., gathers $k$ messages in a given destination node in $O(T_{Lay} + k)$ rounds.
\end{theorem}


The result follows from using the pseudo-BFS layering from \Cref{thm:layer} with the following lemma:

\begin{lemma}\label{thm:gathering}
Given a $C$-collision-free layering $\ell(.)$ with a $C$-coloring $c(.)$, depth $D'$, and source node $s$, \Cref{alg:gatheringCF} gathers $k$ messages in $s$ with high probability, in $C \cdot \Theta(D' + k + \log^2 n)$ rounds. 
\end{lemma}

\begin{algorithm}[t]
\caption{Gathering Algorithm @ node $u$}
\begin{algorithmic}[1]
\normalsize
\Statex Given: Layer $\ell(u)$, color $c(u)$, parent-ID $parent(u)$, a set of initial messages $M$
\Statex Semantics: each packet is 4-tuple in form (message, destination, wave, delay)
\Statex
\State $P \gets \emptyset$
\For {each message $m \in M$}
	\State Choose delay $\delta \in_{\mathcal{U}} [8\max\{2^{-wave} k, 4\log n\}]$
	\State Create packet $\tau \gets (m, parent(u), 0 , \delta)$ and add $\tau$ to $P$
\EndFor
\For {$epoch = 0$ to $\Theta(D' + 16k + \log^2 n)$}\Comment{Main Gathering Part}
	\For {$cycle = 1$ to $C$} 
		\If {$c(u)=cycle$}
		 	\If {$\exists$ exactly one $\pi \in P$ such that $epoch =  D' - \ell(u) + \pi.delay$} 
				\State \textsc{transmit} packet $\pi$
				\State \textsc{listen}
				\If {received acknowledgment}
					\State remove $\pi$ from $P$
				\EndIf
			\Else
				\State \textsc{listen}
				\State \textsc{listen}
			\EndIf	
			\For {$\pi \in P$ s.t. $epoch =  D' - \ell(u) + \pi.delay$} 
				\State Choose random delay $\delta' \in_{\mathcal{U}} [8\max\{k 2^{-wave-1}, 4\log n\}]$
				\State $MaxPreviousDelay \gets  \sum_{1 \leq i \leq wave} 8\max\{k 2^{-i}, 4\log n\}$
				\State remove $\pi$ from $P$
				\State $\pi'\gets (\pi.m,\pi.destination ,\pi.wave + 1, MaxPreviousDelay + \delta')$
				\State add packet $\pi'$ to $P$
			\EndFor
		\Else
			\State \textsc{listen}
			\If {received a packet $\sigma$ such that $\sigma.destination = ID(u)$}
				\State add packet $\sigma'= (\sigma.m,parent(u),\sigma.wave,\sigma.delay)$ to $P$
				\State \textsc{transmit} acknowledgment packet 
			\Else
				\State \textsc{listen}
			\EndIf
		\EndIf
	\EndFor
\EndFor
\label{alg:gatheringCF}
\end{algorithmic}
\end{algorithm}

The full algorithm is presented in \Cref{alg:gatheringCF}. Next, we give an intuitive explanation of its approach. \shortOnly{The formal arguments are deferred to the proof of \Cref{thm:gathering} in the full version.}
Consider the hypothetical scenario where simultaneous transmissions are not lost (no collision) and packet sizes are not bounded, i.e., a node can transmit arbitrary many messages in one round. Consider the simple algorithm where (1) each node $u$ transmits exactly once and in round $D' - \ell(u)$, where it transmits all the messages that it has received by then, (2) a node $v$ accepts a received packet only if $v$ is the parent of the sender. It is easy to see that this is like a wave of transmissions which starts from nodes at layer $D'$ and proceeds exactly towards source, one hop in each round. This wave sweeps the network in a decreasing order of the layer numbers and every message $m$ gets picked up by the wave, when the wave reaches the node that holds $m$ initially. Then, messages are carried by the wave and they all arrive at the source when the wave does, i.e., after exactly $D'$ rounds.  

Things are not as easy in radio networks due to collisions and bounded size messages; each node can only transmit one message at a time, and simultaneous transmissions destined for a common parent collide. We say that \emph{``transmission of message $m$ at node $u$ failed''} if throughout the progress of a wave, message $m$ fails to reach from node $u$ to the parent of $u$ because either (i) a collision happens at $u$'s parent, or (ii) $u$ has other messages scheduled for transmission in the same round as $m$. 
To overcome these, we use two ideas, namely $C$-collision-free layering $\ell()$ with coloring $c()$, and random delays. We use a $C$-collision-free layering by scheduling the transmissions based on colors. This takes care of the possible collisions between nodes of different layer numbers (at the cost of increasing round complexity to $C\cdot D'$). 

Even with the help of a $C$-collision-free layering, we still need to do something for collisions between the transmission of the nodes of the same layer\shortOnly{, and packets scheduled for simultaneous transmission from the same node.}\fullOnly{. Also, note that in the above simple algorithm, messages which their ancestry path to the source goes through a fixed node $v$ are all scheduled for transmissions at the same round in node $v$. This obviously results in transmission failures of type (ii).} The idea \fullOnly{to get over these two \emph{transmission failure} origins} is to add a random delay to the transmission time of each message. If there are $k$ active messages and we add a random delay chosen from $[8k]$ to each message, then for each message $m$, with probability at least $7/8$ no transmission of $m$ fails, i.e., the wave delivers $m$ to the source with probability at least $7/8$. A formal argument for this claim would be presented in the proof. With this observation, one naive idea would be to repeat the above algorithm on a $C$-collision-free layering $\ell()$, by having $\Theta(\log n)$ non-overlapping waves, where each time each message starts from the node that it got stuck in while being carried by the previous wave. With this, we succeed with high probability in delivering all $k$ messages to the source and in time $C \cdot O(D'\log n+ k\log n)$. 

Now there are two ideas to improve upon this. First, we can \emph{pipeline} the waves. That is, we do not need to space the waves $D'$ rounds apart; instead the spacing should be just large enough so that two waves do not collide. For that, a spacing of $8k$ between the waves is enough. With this improvement, we go down to time complexity of $C \cdot O(D' + k\log n)$.    
Second, note that in each wave, each message succeeds with probability at least $7/8$. Thus, using Chernoff bound, we get that as long as the number of remaining messages is $\Omega(\log n)$, whp, in each wave, the number of remaining messages goes down by at least a $\frac{1}{2}$ factor. Hence, in those times, we can decrease the size of the interval out of which the random delays are chosen by a factor of two in each new wave. Because of this, the spacing between the waves also goes down exponentially. This second improvement, with some care for the case where number of remaining messages goes below $\Theta(\log n)$ (where we do not have the Chernoff-type high probability concentration anymore) gives time complexity of $C \cdot O(D' + k + \log^2 n)$.     

\fullOnly{
Next we give the formal proof of \Cref{thm:gathering} which now be easy to understand:

\begin{proof}[Proof of \Cref{thm:gathering}] We first argue that each packet gets routed to the root of the layering eventually. In the absence of collisions this is true because the layer number of a receiving parent $w$ is always smaller than the one of the sender $u$ of a packet this node will retransmit the packet later to its parent. If on the other hand a collision prevents the parent $w$ from receiving a packet then $w$ will not acknowledge this packet to $u$ and $u$ will pick a new larger random delay for this packet and try again later. It is also good to see that packets do not get duplicated which would happen if packets arrive but their acknowledgments collide. This is not possible since if two acknowledgments from nodes $w$ and $w'$ collide at a node $u$ one of them must be for a transmission that came not from $u$ but all nodes connected to $u$ will either have received its message or a collision in the round before. 

With this in mind it is clear that at any point of time there is at most one node per message $m$ that is trying to send $m$ in a packet. Since we schedule transmissions according to the colors, in a $C$-collision-free layering we get the advantage that only transmissions from nodes in the same layer can interfere. That happens only if packets have the same delay value. Furthermore, the ranges of delay values that a node can have do not overlap. This guarantees that each packet might have conflict only with packets in the same wave. We show that in each wave, each message has an independent probability of at least $1/2$ to be collision-free. This shows that with high probability, all messages are delivered after at most $4\log n$ waves. Thus, considering the values of delays at each wave, we get that each message is with high probability delivered to the source after at most $C \cdot\Theta(D'+ k + \log^2 n)$ rounds.

In order to show that in each wave, each message is delivered to the source with probability at least $1/2$, we show by induction that the number of active packets is at most $1/8$ times the size of the range from which the random additional delays $\delta'$ are chosen. This is true in the beginning. Furthermore, this implies that if one fixes the delay choices of all packets except for one, at least $7/8$ random delay values will not result in a collision for this packet. This implies that each packet gets independently delivered with probability at least $7/8$. If $k 2^{-i} > 4\log n$, then in wave $i$, with high probability at least half of the messages succeed in being delivered to the source and thus, do not participate in the next wave. When $k 2^{-i} < \log n$, the delay values are at least $8 \cdot 4 \log n$ large, because of the max-expressions in lines 3 and 18. This proves the inductive step, completing the whole proof.
\end{proof}
}

\section{Multi-Message Broadcast, and Gossiping}\label{sec:bcast}

\begin{algorithm}[t]
\caption{Network-Coded Multi-Message Broadcast @ node u}
\begin{algorithmic}[1]
\normalsize
\Statex Given: Source node $s$ with $k$ messages
\Statex
\If{$u=source$}
	\For{ all $i \in [k]$}
		\State $v_i \gets (e_i,m_i)$    \Comment{$e_i \in \{0,1\}^k$ is the $i^{th}$ basis vector}
		\State put $v_i$ in $P$
	\EndFor
\Else
	\State $P \gets \emptyset$	
\EndIf
\Statex
\For{ $i = 1$ to $\Theta(D' \log \frac{n}{D'} + k \log n + \log^2 n)$}
	\For{ $cycle = 1$ to $C$}
		\If {$cycle \equiv c(u)$}
			\WithProb {$2^{-BC_{i \mod L}}$} \Comment{$BC$ is the Broadcast sequence from \Cref{sec:decaynew}}
				\State choose a uniformly random subset $S$ of $P$
				\State \textsc{transmit} $\bigoplus_{v \in S} v$
			\Otherwise
				\State \textsc{listen}
			\EndWithProb
		\Else
			\State \textsc{listen}
		\EndIf
		\If {received a packet $v$} add $v$ to $P$
		\EndIf	
	\EndFor
\EndFor
\Statex
\State decode $v_1, \ldots, v_k$ from $span(P)$ by Gaussian Elimination

\label{alg:NCBcast}
\end{algorithmic}
\vspace{-0.1cm}
\end{algorithm}

In this section we show how to combine psuedo-BFS layerings, the broadcat protocol of \Cref{sec:decaynew}, and the idea of random linear network coding to obtain a simple and optimal $O(T_{Lay} + k \log n)$ $k$-message broadcast algorithm. 
Note that the $\Omega(k \log n)$ lower bound of \cite{NCLB}, along with the $\Omega(D\log \frac{n}{D}+\log^2n)$ broadcast lower bound of \cite{KM} and \cite{ABLP}, show the near optimality of this algorithm.

\begin{theorem}\label{thm:Bcast} Given $k$ messages at a single source, there is a randomized distributed algorithm that broadcasts these $k$ messages to all nodes, w.h.p, in $O(T_{Lay} + k \log n)$ rounds.
\end{theorem}

The result follows from using the pseudo-BFS layering from \Cref{thm:layer} with the following lemma:

\begin{lemma}\label{lem:codedBC}
Given a $C$-collision free layering $\ell$ with depth $D'$ and $k$ messages at source $s$, the Network-Coded Multi-Message Broadcast algorithm delivers all messages to all nodes, w.h.p., in $C \cdot O(D' \log \frac{n}{D'} + k \log n + \log^2 n)$ rounds. 
\end{lemma}


The algorithm is presented in \Cref{alg:NCBcast}. The main ideas are as follows. To schedule which node is sending at every time, we first restrict the nodes that are sending simultaneously to have the same color. To resolve the remaining collisions, we let nodes send independently at random with probabilities chosen according to the CR-Broadcast protocol of \cite{CR} with parameter $\delta=\log \frac{n}{D'}$. Lastly, if a node is prompted to send a packet, we create this packet using the standard distributed packetized implementation of random linear network coding as described in \cite{Haeupler11}. Given such a random linear network code, decoding can simply be performed by Gaussian elimination (see \cite{Haeupler11}).

The proof uses several ideas stemming from recent advances in analyzing random linear network coding. The key part is the \emph{projection analysis} of \cite{Haeupler11} and its modification and adaption to radio networks \cite{BCSTCD}, titled \emph{backwards projection analysis}. This allows us to reduce the multi-message problem to merely showing that, for each particular node $v$, one can find a path of successful transmissions from the source to $v$ with exponentially high probability. The required tail-bound follows from a slightly modified analysis of the CR-Broadcast protocol \cite{CR}. We remark that the additive coefficient overhead in \Cref{alg:NCBcast}, which is one-bit for each of the $k$ messages, can be reduced to $O(\log n)$ bits using standard techniques explained in \cite{BCSTCD}. \shortOnly{The proof of \Cref{lem:codedBC} appears in the appendix.} 

Before going into details of the proof of \Cref{lem:codedBC}, we present a Chernoff-type concentration bound. 
This lemma is later used to bound the tail of the probability that there is a path of successful transmissions from the source to a particular node.

\begin{lemma}\label{lem:tailbound}
For any $D', n, k$ let $s_1,\ldots,s_{D'}$ be integers between $0$ and $n$ and let $X_i$ with $i \in [D']$ be i.i.d. geometric random variables with success probability $p$. We have:
$$P[ \sum_{i \in [D']} s_i X_i > \frac{2}{p}(\frac{1}{p} \sum_{i \in [D']} s_i + k/\max_i\{s_i\}) ] < (1 - p)^k$$
\end{lemma} 

\begin{proof}[Proof of \Cref{lem:codedBC}]

We interpret all messages and all packets used in the algorithm as (bit) vectors over the finite field $GF(2)$. With this, all packets created in step $3$ have the form $(\mu,m_\mu) = (\mu_1,\ldots,mu_k,\sum_{i \in [k]} \mu_i m_i) \in GF(2)^{k+l}$ where $l$ is the size of a message. Since we only XOR (or equivalently add) packets of this form during the algorithm and since $(\mu,m_\mu) \oplus (\mu',m_{\mu'}) = (\mu + \mu',m_{\mu + \mu'})$ this invariant is preserved throughout. Also, if a node receives $k$ messages $(\mu_1,m_{\mu_1}), \ldots, (\mu_k,m_{\mu_k})$ in which the vectors $\mu_1,\ldots, \mu_k$ are independent and span the full $k$ dimensional space $GF(2)^k$ then all messages can be recovered by Gaussian elimination. This allows us to solely concentrate on the spreading of the coefficient parts through the network. The goal of the rest of this proof is thus to show that these vectors spread such that in the end all nodes receive the full coefficient space with high probability.

Instead of tracking the coefficient vectors themselves, we follow \cite{Haeupler11} and look at their projections. More precisely, we say a node $u$ knows about a projection vector $\mu \in GF(2)^k$ if the projection of its received packets onto $\mu$ is non-zero, that is, if $u$ has at least one packet $(\mu',m_{\mu'}) \in P_u$ with a non-perpendicular coefficient vector $\mu'$ (i.e, $\left\langle \mu, \mu' \right\rangle \neq 0$). Our main claim is that for every node $u$ and every projection vector $\mu$ the probability that after $T = \Theta(D' \log \frac{n}{D'} + k \log n + \log^2 n)$ rounds node $u$ does not know about $\mu$ is at most $2^{-(k + 2\log n)}$. A union bound over all $2^k$ vectors and all $n$ nodes then shows that with high probability every node knows about every projection vector. From this one can easily conclude that every node can decode all messages due to having received vectors that span the full coefficient space (a lower dimensional span would directly give a perpendicular and thus unknown projection).

To prove the main claim we focus on one node $u$ and one projection vector $\mu$. With $\mu$ fixed we define a transmission to be a $\mu$-failure iff the node node sending it knows $\mu$ but the packet in the transmission carries a coefficient vector that is perpendicular to $\mu$. It is easy to see that in order for node $u$ to know $\mu$ in the end it is necessary and sufficient that there is a sequence of transmissions starting from the source $s$ and ending at $u$ in which each transmission is both collision free and not a $\mu$-failure. Furthermore, these two types of failures, collisions and $\mu$-failures, are independent in the sense that any collision free transmission is a $\mu$-failure with an independent probability of at most $1/2$: If, on the one hand, the sender $u$ in a transmission has zero packets in $P_u$ that are non-perpendicular to $\mu$ then it does not know $\mu$ and cannot fail a transmission. On the other hand, if the receiver has at least one packet in $P$ with a coefficient vector that is non-perpendicular to $\mu$ then the probability that the parity of the number of these packets that are included in $S$ is even, which is what is needed for a $\mu$-failure, is exactly $1/2$. We call a transmission successful if it is both collision free and not a $\mu$-failure. 

To find the desired transmission sequence we employ the backward analysis introduced by Ghaffari et al. in~\cite{BCSTCD}. For this, instead of constructing a sequence from $s$ to $u$, we go backwards in time starting at round $T$ and try to find a transmission sequence from $u$ to $s$. In particular, for $t$ decreasing from $T$ to $1$ we say that at time $t$ we have progressed up to layer $l$ if $l$ is the smallest layer number such that there is a node $v$ at layer $l = l(v)$ for which there exists a sequence of successful transmissions from $v$ to $u$ between time $t$ and $T$. 

Suppose at time $t$ we have progressed up to node $v$ in layer $l = l(v)$. We want to analyze how long it takes for the next step. For this, we look at node $v$ and fix $w$ to be a neighboring node with a smaller layer number. The existence of such a node is guaranteed by the layering property in \Cref{def:layering}. Now, let $N$ be the set of neighbors of $u$ with the same color as $w$. Note, that because the layering is $C$ collision free these neighbors also have the same layer number as $w$. In \cite{CR} it was proven that over the course of $FD(|N|)$ rounds there is a constant probability for $v$ to receive a transmission from a node in $N$. Here


$$FD(|N|) = \threepartdefotherwise{3\log \frac{n}{D'}}{|N| < \frac{n}{D'}}{3|N|\frac{D'}{n}\log {\frac{n}{D'}}}{\frac{n}{D'} < |N| < \frac{n \log n}{D'}}{3\log n}$$
comes from the definition of the CR-Broadcast protocol in \Cref{sec:decaynew}. Lastly, as explained before, there is at most an independent $1/2$ probability that such a collision free transmission is a $\mu$-failure. In total we get that over the course of at most $s_l = FD(n_l)$ rounds, where $n_l = |\{w'|l(w')=l\}|$, there is an independent chance of at least $1/4e$ that progress to layer $l(w) <l$ is made. Given this, the total time needed to progress from $v$ to the source $s$ is thus dominated by $\sum_{l \in [D']} s_l X_l$ where $X_i$ for all $i \in [D']$ are i.i.d. geometric random variables with success probability $1/4e$. For this setting \Cref{lem:tailbound} shows that the probability that after $T = 8e(\sum_{l \in [D']} s_l + (2\log n + k))$ rounds no sequence of successful transmissions from $s$ to $v$ exists is at most $(1 - 1/4e)^{5e(2\log n + k)} < 2^{- (2\log n + k)}$. This is precisely the probability bound promised in the main claim given that $$\sum_{l \in [D']} s_l \leq D' (3 \log \frac{n}{D'}) + \sum_{l \in [D'], s_l \geq \log 3\frac{n}{D'}} 3|n_l|\frac{D'}{n} \geq  3D' \log \frac{n}{D'} + 3 D' \frac{\sum_l n_l}{n} \leq  3D' \log \frac{n}{D'} + 3 D',$$
where $\sum_l n_l < n$ comes from the disjointness of layers. This shows that $T = \Theta(D' \log \frac{n}{D'} + k \log n + \log^2 n)$ rounds suffice with high probability as claimed.  

We remark that the additive coefficient overhead in the messages in \Cref{alg:NCBcast}, which is one-bit for each of the $k$ messages, can be reduced to $O(\log n)$ bits using standard techniques explained in \cite{BCSTCD}. For this we reuse a schedule for $k = \Theta(\log n)$ messages via pipe-lining. That is, we run a network coded broadcast of $k = \log n$ messages but cut it into phases of $\Theta(\log^2 n)$ rounds. If a node has not received all messages at the end of a phase, it empties its buffer $P$ and restarts in the next phase. From the proof of \Cref{lem:codedBC}, we get that, whp, the $k$ messages proceed in each phase at least as far as the next few levels whose $s_l$ values sum up to $\Theta(\log^2 n)$. In total, $\Theta(D' \log \frac{n}{D'} + \log^2 n)$ rounds still suffice to spread the $k = \log n$ messages. Now, we can repeat the schedule while reusing the same transmission schedule and coding coefficients for every block of $\log n$ messages.
\end{proof}

\smallskip 

Lastly, we present our gossiping result.

\begin{theorem}\label{thm:gossip} There is a randomized distributed algorithm that, with high probability, performs an all-to-all broadcast in $O(n\log n)$ rounds.
\end{theorem}
\begin{proof} First, we elect a leader node in $O(n)$ rounds using the algorithm of \cite{SODA-LE}. Then, we construct a pseudo-BFS layering around this leader in time $O(n)$ using \Cref{thm:layer}. We then gather the $n$ messages in the leader node in $O(n)$ rounds using \Cref{thm:gathering}. Finally, we broadcast the $n$ messages from the leader to all the other nodes in time $O(n\log n)$ using \Cref{lem:codedBC}.
\end{proof}


Similar to the approach of the proof of \Cref{thm:gossip} one can also combine the leader election algorithm of \cite{SODA-LE} with the pseudo-BFS layering, gathering, and single-source broadcast algorithms of this paper and obtain a near optimal randomized distributed algorithm for the multi-source version of \Cref{thm:Bcast}, that is, a multi-source $k$-message broadcast. 

\begin{theorem}\label{thm:multisource} 
Given $k$ messages at different sources, there is a randomized distributed algorithm that broadcasts these $k$ messages to all nodes, with high probability, in $O((D\log{\frac{n}{D}}+\log^3 n)\cdot \min\{\log\log n, \log{\frac{n}{D}}\} + k\log n)$ rounds.
\end{theorem}

\bibliographystyle{acm}
\bibliography{Bdata}
\end{document}